\documentclass{llncs}

\usepackage{graphicx}
\usepackage{picins}
\usepackage{amsmath}
\usepackage{amssymb}
\usepackage{epsfig}
\usepackage[usenames,dvipsnames]{color}
\usepackage{thm-restate}

\begin{document}

\title{Optimal Grid Drawings of Complete Multipartite Graphs and an Integer Variant of the Algebraic Connectivity }

\author{Ruy Fabila-Monroy, Carlos Hidalgo-Toscano, Clemens Huemer, Dolores Lara, Dieter Mitsche}

\institute{ Laboratoire Dieudonné, Univ. Nice, dmitsche@unice.fr,}

\author{%
Ruy Fabila-Monroy\inst{1}  \and 
Carlos Hidalgo-Toscano\inst{1}
Clemens Huemer\inst{2} \and 
Dolores Lara\inst{1} \and 
Dieter Mitsche\inst{3}
}%
\institute{
CINVESTAV-IPN,\\
\email{ruyfabila@math.cinvestav.edu.mx, cmhidalgo@math.cinvestav.mx,  dlara@cs.cinvestav.mx},\\ 
\and
Universitat Polit\`ecnica de Catalunya,\\
\email{clemens.huemer@upc.edu}
\and
Laboratoire Dieudonn\'e, Univ. Nice\\
\email{dmitsche@unice.fr}
}

\maketitle

\begin{abstract}
How to draw the vertices of a complete multipartite graph $G$ on different points of a bounded $d$-dimensional integer grid, such that the sum of squared distances between vertices of $G$ is (i) minimized or (ii) maximized? For both problems we provide a characterization of the solutions.
For the particular case $d=1$, our solution for (i) also settles the minimum-2-sum problem for complete bipartite graphs; the minimum-2-sum problem was defined by Juvan and Mohar in 1992. Weighted centroidal Voronoi tessellations are the solution for (ii). Such drawings are related with Laplacian eigenvalues of graphs. This motivates us to study which properties of the algebraic connectivity of graphs carry over to the restricted setting of drawings of graphs with integer coordinates.

\end{abstract}

\section{Introduction}

Let $r,d$ be positive integers. Let $n_1 \le \dots \le n_r$ be positive integers such
that $\sum n_i= (2M+1)^d$ for some integer $M$. We consider straight line drawings
of the complete $r$-partite graph $K_{n_1,\dots,n_r}$ into the $d$-dimensional integer grid
\[P:=\left \{ (x_1,\dots,x_d) \in \mathbb{Z}^d:-M \le x_i \le M \right \}.\]
No two vertices of the graph are drawn on the same grid point. %, and all the edges are drawn as straight line segments.
Note that such a drawing corresponds to a coloring of the points of $P$ with $r$ colors, such that color $i$ appears $n_i$ times, for $i=1,\ldots,r$.
The goal is to find the assignment of colors to the points of $P$ such that the sum of squared distances between points of different colors is (i) minimized or (ii) maximized. The motivation for this problem stems from the following relation between drawings of a graph and spectral theory:

Let $G=(V,E)$ be a graph with vertex set $V=\{1,\dots,N\}$, and let $deg(i)$ denote the degree of vertex $i$. The \emph{Laplacian}
matrix of $G$ is the $N \times N$ matrix, $L=L(G)$, whose entries are
\begin{equation*}
 L_{i,j} = \left\{ \begin{array}{cl}
         deg(i), & \text{if } i=j,\\
        -1, & \text{if } i \neq j \text{ and } ij \in E,\\
        0, & \text{if } i \neq j \text{ and } ij \notin E.\end{array} \right.
\end{equation*}
Let $\lambda_1(G) \le \lambda_2(G) \le \dots \le \lambda_N(G)$ be the eigenvalues of $L$.
The \emph{algebraic connectivity} (also known as the Fiedler value~\cite{fiedler}) of $G$
is the value of $\lambda_2(G)$. It is related to many graph invariants (see~\cite{fiedler}), and in particular to the size of the separator of a graph, giving rise to partitioning techniques using the associated eigenvector (see~\cite{spielman}).
 Spielman and Teng~\cite{spielman} proved the following lemma:
\begin{lemma}[Embedding Lemma]\label{lem:emb}
\[ \lambda_2(G) = \min \frac{\sum_{ij \in E} \|\vec{v}_i-\vec{v}_j\|^2}{\sum_{i \in V} \|\vec{v}_i\|^2}, \]
 and
\[ \lambda_N(G) = \max \frac{\sum_{ij \in E} \|\vec{v}_i-\vec{v}_j\|^2}{\sum_{i \in V} \|\vec{v}_i\|^2}, \]
where the minimum, respectively maximum, is taken over all tuples $(\vec{v}_1,\dots,\vec{v}_N)$ of vectors 
$\vec{v}_i \in \mathbb{R}^d$ with $\sum_{i=1}^{N} v_i =\mathbf{0}$, and not all $v_i$ are zero-vectors $\mathbf{0}$.
 %\flushright{\QED}
\end{lemma}
In fact, Spielman and Teng~\cite{spielman} proved the Embedding Lemma for $\lambda_2(G)$, but the result
for $\lambda_N(G)$ follows by very similar arguments; when adapting the proof of~\cite{spielman} we have to replace 
the last inequality given there by the inequality $\sum_i x_i / \sum_i y_i \leq \max_i \frac{x_i}{y_i}$, for $x_i, y_i >0.$

Let $\mathbf{v}= (\vec{v}_1,\dots,\vec{v}_N)$ be a tuple of positions
defining a drawing of $G$ (vertex $i$ is placed at $\vec{v}_i$).
Let
\begin{equation}\label{eq:def}
\lambda(\mathbf{v}):=\frac{\sum_{ij \in E} \|\vec{v}_i-\vec{v}_j\|^2}{\sum_{i \in V} \|\vec{v}_i\|^2}.
\end{equation}

Note that $\|\vec{v}_i-\vec{v}_j\|^2$ is equal to squared length of the
edge $ij$ in the drawing defined by $\mathbf{v}$.
Lemma~\ref{lem:emb} provides a link between
the algebraic connectivity of $G$ and its straight line drawings.
Clearly
\[\lambda_2(G) \le \lambda(\mathbf{v}) \le \lambda_N(G).\]
We remark that in dimension $d=1$, optimal drawings $\mathbf{v}$ are eigenvectors of $L(G)$ and $\lambda(\mathbf{v})$ is the
well known Rayleigh quotient.

In this paper we study how well we can
approximate $\lambda_2(G)$ and $\lambda_N(G)$ with drawings with certain restrictions.
First, we restrict ourselves to drawings in which the vertices are placed at points
with integer coordinates and no two vertices are placed at the same point.
Since $\lambda(\alpha \mathbf{v})=\lambda(\mathbf{v})$ for $\alpha \in \mathbb{R} \setminus \{0\}$, we have that
$\lambda_2(G)$ and $\lambda_N(G)$
can be approximated arbitrarily closely with straight line drawings with integer coordinates of sufficiently
large absolute value. We therefore bound the absolute value of such drawings and consider
only drawings in the bounded $d$-dimensional integer grid $P$.
 Juvan and Mohar~\cite{juvan2,juvan} already studied drawings of graphs with integer coordinates for $d=1$. More precisely, the authors consider 
the {\it{minimum-$p$-sum}}-problem: for $0 < p < \infty$, a graph $G$ and a bijective mapping $\Psi$ from $V$ to $\{1,\ldots,N\}$,
define $\sigma_p(G, \Psi)=\left(\sum_{uv \in E(G)} |\Psi(u)-\Psi(v)|^p \right)^{1/p}$, 
and for $p=\infty$, let $\sigma_p(G, \Psi)=\max_{uv \in E(G)}|\Psi(u)-\Psi(v)|$. The quantity $\sigma_p(G)=\min_{\Psi} \sigma_p(G,\Psi)$ 
(where the minimum is taken over all bijective mappings) is then called the minimum-$p$-sum of $G$, and if $p=\infty$, it is also called 
the bandwidth of $G$. In~\cite{juvan} relations between the min-$p$-sum and $\lambda_2(G)$ and $\lambda_N(G)$ are analyzed, and also 
polynomial-time approximations of the minimum-$p$-sum based on the drawing suggested by the eigenvector corresponding to $\lambda_2(G)$ are
given. In~\cite{juvan2}  the minimum-$p$-sums and its relations to $\lambda_2(G)$ and $\lambda_N(G)$ are studied for the cases of
random graphs, random regular graphs, and Kneser graphs. For a survey on the history of these problems, see~\cite{chinn} 
and~\cite{chvatalova}.\\
The use of eigenvectors in graph drawing has been studied for instance in~\cite{koren}, and we also mention~\cite{rocha} as a recent work on spectral bisection.\\

In the next two sections we characterize the optimal drawings $\mathbf{v}$ for complete multipartite graphs $K_{n_1,\ldots,n_r}$ which minimize/maximize $\lambda(\mathbf{v})$. The assumption $N=\sum_{i=1}^{r}n_i = (2M+1)^d$ made in the beginning is to ensure that every drawing
satisfies the condition $\sum_{i=1}^N \vec{v}_i=\mathbf{0}$.
The Laplacian eigenvalues of $K_{n_1,\ldots,n_r}$ are known to be, see~\cite{bolla},
$$0^1, (N-n_r)^{n_r-1}, (N-n_{r-1})^{n_{r-1}-1},\ldots, (N-n_2)^{n_2-1}, (N-n_1)^{n_1-1},N^{r-1}$$
where the superindexes denote the multiplicities of the eigenvalues.
% and $0$. %with multiplicities $k-1$, $n_i-1$ and $1$ respectively.
Therefore, $N-n_r \le \lambda(\mathbf{v}) \le N.$\\

\begin{figure}[t]
\begin{minipage}[b]{0.45\linewidth}
\centering
\includegraphics[scale=0.2]{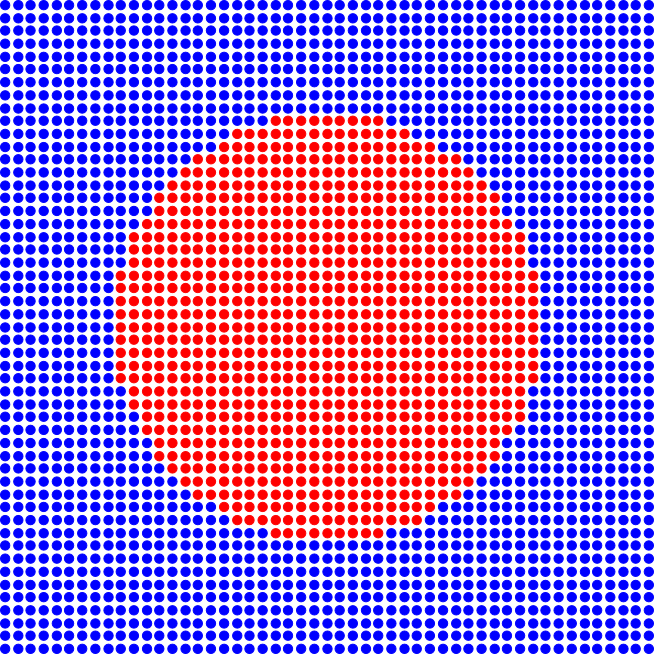}
\end{minipage}
\hspace{0.5cm}
\begin{minipage}[b]{0.45\linewidth}
\centering
\includegraphics[scale=0.2]{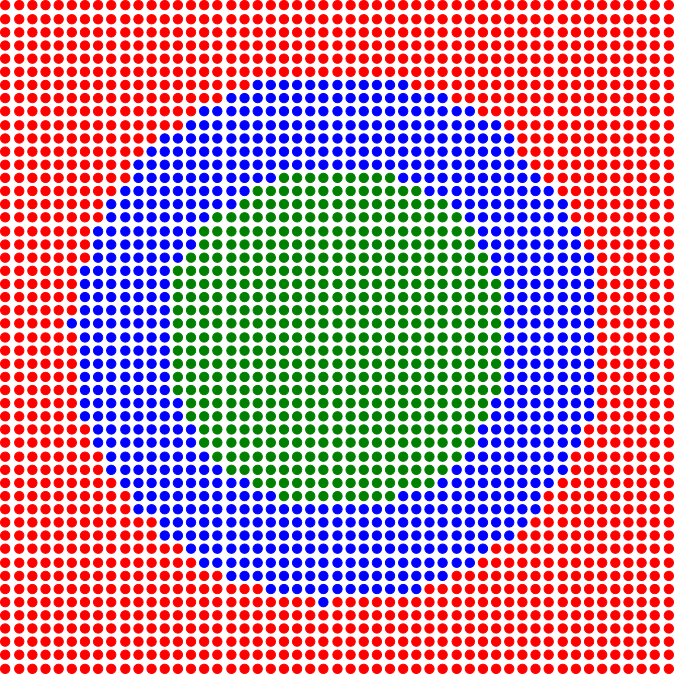}
\end{minipage}
\caption{The best way to minimize the sum of squared distances between points of different colors on a $51 \times 51$ integer grid. Left: for $r=2$ colors with $1/3$ of the points in red and $2/3$ of the points in blue. Right: for $r=3$ colors, with $1359$ red points, $724$ blue points, and $518$ green points. }
\label{fig:minimize2colors}
\end{figure}

Two examples of optimal drawings in dimension $d=2$ which minimize $\lambda(\mathbf{v})$ are given in Figure~\ref{fig:minimize2colors}.
Figure~\ref{fig:maximize2colors} and Figure~\ref{fig:maximize6colors} show examples which maximize $\lambda(\mathbf{v})$. 
We mention that we obtained all these drawings with computer simulations, using simulated annealing.
The solution for minimizing $\lambda(\mathbf{v})$ shown in Figure~\ref{fig:minimize2colors} consists of concentric rings and applies to the case when 
all color classes have different size. While this solution is unique, we will show that if the color classes have the same size, 
then there are exponentially many drawings that minimize $\lambda(\mathbf{v})$. In that case, the solutions are characterized as those drawings 
where for each color class 
all its points sum up to $\mathbf{0}$. 
\begin{figure}[t]
\begin{minipage}[b]{0.45\linewidth}
\centering
\includegraphics[scale=0.2]{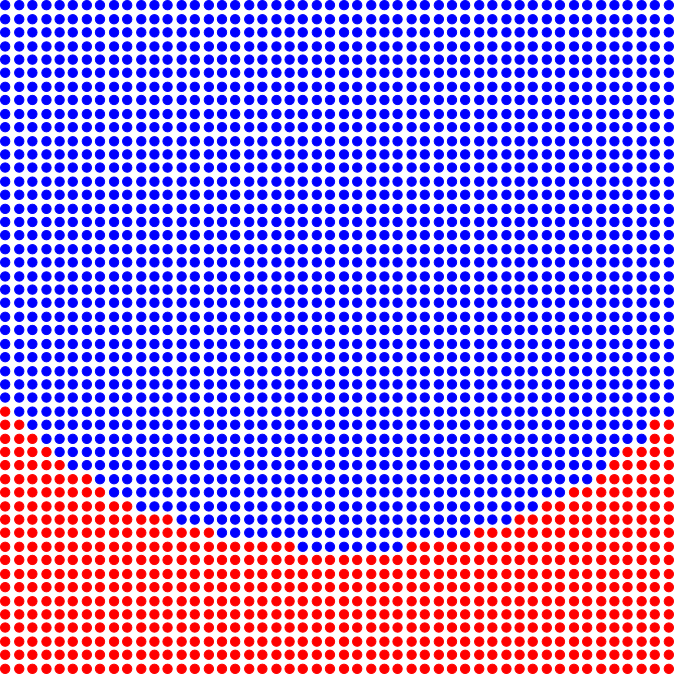}
\end{minipage}
\hspace{0.5cm}
\begin{minipage}[b]{0.45\linewidth}
\centering
\includegraphics[scale=0.22]{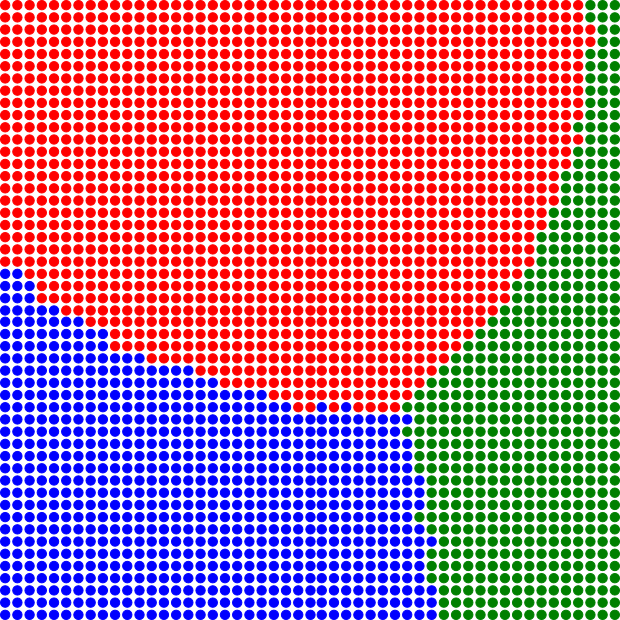}
\end{minipage}
\caption{The best way to maximize the sum of squared distances between points of different colors on a $51 \times 51$ integer grid. Left: for $r=2$ colors with $3/4$ of the points in blue and $1/4$ of the points in red. Right: for $r=3$ colors, with $1359$ red points, $724$ blue points, and $518$ green points. }
\label{fig:maximize2colors}
\end{figure}
\begin{figure}[t]
\begin{minipage}[b]{0.45\linewidth}
\centering
\includegraphics[scale=0.2]{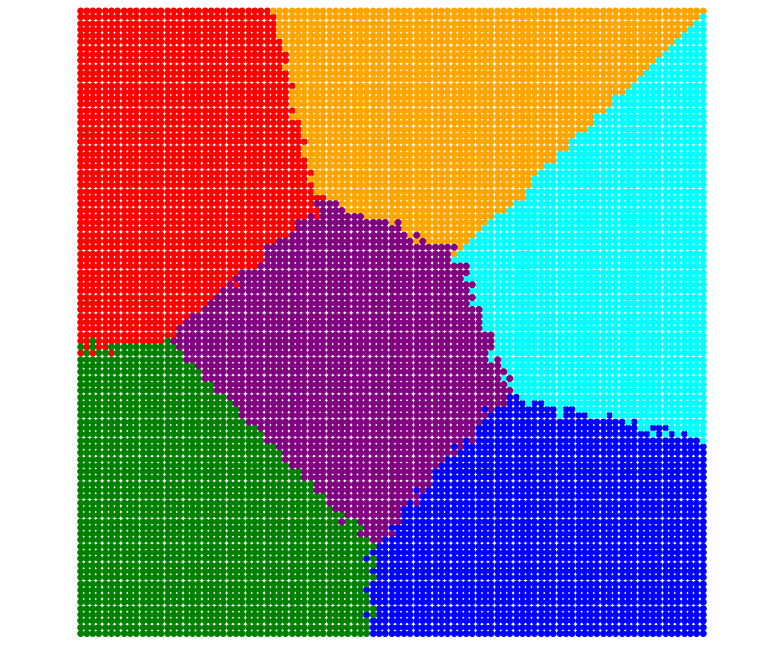}
\end{minipage}
\hspace{0.5cm}
\begin{minipage}[b]{0.45\linewidth}
\centering
\includegraphics[scale=0.2]{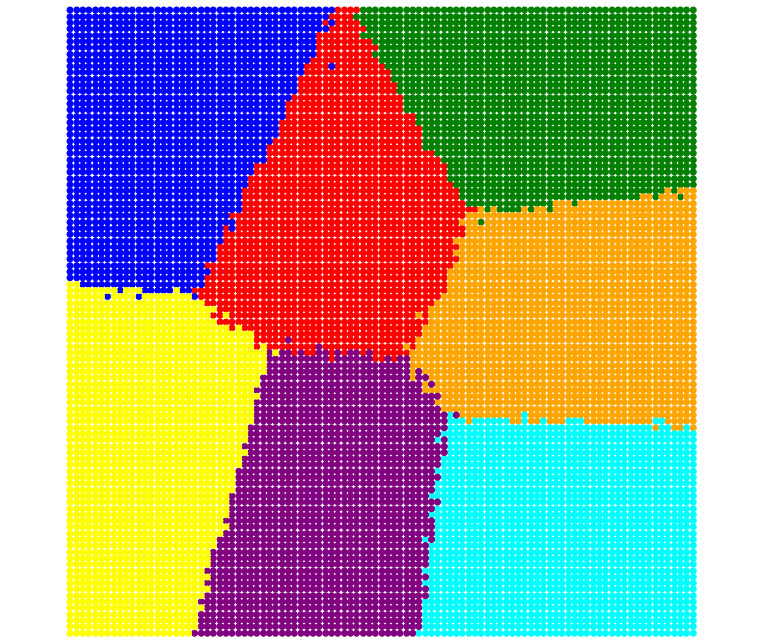}
\end{minipage}
\caption{The best way to maximize the sum of squared distances between points of
	different colors on a $101 \times 101$ integer grid. Left: for $r = 6$ colors, with $1701$ purple points and $1700$ points of every other color. Right: for $r = 7$ colors, with $1459$ yellow points and $1457$ points of every other color.}
%\caption{The best way to maximize the sum of squared distances between points of
%different colors on a $51 \times 51$ integer grid. Left: for $r = 6$ colors with same number of
%points per color. Right: for $r = 7$ colors, with same number of points per color.}
\label{fig:maximize6colors}
\end{figure}
As can be observed in Figures~\ref{fig:maximize2colors} and Figure~\ref{fig:maximize6colors} (and proved in Section 3), 
the solution for maximizing $\lambda(\mathbf{v})$ is given by (weighted) centroidal Voronoi diagrams, 
which are related to clustering~\cite{du}. Let us give the definition of a centroidal Voronoi tessellation, according to~\cite{du}. 
Given an open set  $\Omega \subseteq \mathbb{R}^d$, the set $\{V_i\}_{i=1}^{r}$ is called a {\it{tessellation}} of $\Omega$ if $V_i \cap V_j = \emptyset$ for $i\neq j$
and $\cup_{i=1}^{r} \overline{V_i} = \overline{\Omega}.$ Given a set of points $\{c_i\}_{i=1}^{r}$ belonging to  $\overline\Omega $, the Voronoi region $\hat{V_i}$
corresponding to the point $c_i$ is defined by 
$$\hat{V_i}=\{x \in \Omega \ | \ ||x-c_i|| < || x - c_j|| \ \mbox{for} \ j=1,\ldots,r, j \neq i\}.$$
The points $\{c_i\}_{i=1}^{r}$ are called {\it{generators}} or {\it{sites}}. 
The set $\{\hat{V_i}\}_{i=1}^{r}$ is a {\it{Voronoi tessellation}} or {\it{Voronoi diagram}}.  
A Voronoi diagram is {\it{multiplicatively weighted}}, see~\cite{auren}, if each generator $c_i$ has an associated weight $w_i>0$ and the weighted Voronoi region of $c_i$ is
$$\hat{V_i}=\{x \in \Omega \ | \ ||x-c_i||w_j < || x - c_j||w_i \ \mbox{for} \ j=1,\ldots,r, j \neq i\}.$$
A Voronoi tessellation is {\it{centroidal}} 
if the generators are the centroids for each
Voronoi region. 
Voronoi diagrams have also been defined for discrete sets $P$ instead of regions $\Omega$~\cite{du}.\\

Finally, in Section 4 we focus on graph drawings in dimension $d=1$ and treat the question on what can be said about approximations of 
eigenvectors with bounded integer vectors. In particular, we
study the relation between the algebraic connectivity and an integer version of the algebraic connectivity and the minimum-2-sum. We think analogous relations should
also hold for drawings in higher dimension; we leave this for further research.

\section{Optimal Drawings for Minimizing $\lambda(\mathbf{v})$ }{\label{sec:minimize}}

In the following we give bounds on $\lambda(\mathbf{v})$. Note that in Equation~(\ref{eq:def}), the term $\sum_{i \in V}\|\vec{v}_i\|^2$ is the same for all drawings
on $P$.
Let  $S:=\sum_{v \in P}  \|v \|^2.$ We first calculate the value of $S$ which we need later on.
\begin{restatable}{proposition}{propS}\label{lem:S}
\[S=2d(2M+1)^{d-1}\frac{M (M+1) (2M+1) }{6}.\]
 \end{restatable}

The proof is given in the appendix.\\

Let $A_1,\dots, A_r$ be the partition classes of $K_{n_1,\dots,n_r}$ with
$|A_i| =n_i$ (for $1 \le  i \le r$). Let $N=(2M+1)^d=\sum_{i=1}^{r} n_i$. 
Let $\mathbf{v}$ be a fixed straight line drawing of $K_{n_1,\dots, n_r}$. In what
follows we abuse notation and say that a point $v \in P$ is in $A_i$
if a vertex of $A_i$ is mapped to $v$. We also use $A_i$ to refer
to the image of $A_i$ under $\mathbf{v}$.

Let $A$ and $B$ be two finite subsets of $\mathbb{R}^d$. We define
\[A \cdot B :=\sum_{\substack{v \in A \\ w \in B }} v \cdot w,\]
where $\cdot$ is the dot product. 
We will need the following property:
\begin{restatable}{proposition}{propProd}{\label{prop:nonpos}}
Let $A_1,\ldots,A_r$ be $r \geq 2$ finite subsets of $\mathbb{R}^d$ such that \[\sum_{i=1}^{r}\sum_{\substack{v \in A_i }} v =\vec{0}.\]
Then 
\[ \sum_{i=1}^{r-1}\sum_{j=i+1}^{r}A_i \cdot A_j = -\frac{1}{2}\sum_{i=1}^{r}\left|\left|\sum_{\substack{v \in A_i }} v\right|\right|^2.\] 
\end{restatable}

The proof is given in the appendix.

\begin{lemma}\label{lem:minvalue}
 Let $\mathbf{v}$ be a fixed straight line drawing of $G=(V,E)=K_{n_1,\dots,n_r}$.% and let $N=\sum_{i=1}^{r} n_i$.
 Then
 \[ \lambda(\mathbf{v})=N+\frac{1}{S} \sum_{i=1}^r \left (  -n_i \sum_{v \in A_i} \|v \|^2\right) +\frac{1}{S} \sum_{i=1}^{r} \left|\left|\sum_{\substack{v \in A_i }} v\right|\right|^2 \nonumber.\]
\end{lemma}
\begin{proof}

$$
 \lambda(\mathbf{v}) =\frac{1}{S}\sum_{(v,w) \in E} \| v-w\|^2  =  \frac{1}{S} \sum_{(v,w) \in E} \left ( \|v \|^2+\|w \|^2-2 v \cdot w  \right ).\\
$$
Since in the complete multipartite graph each $v \in A_i$ is adjacent to all vertices but the $n_i$ vertices of its class $A_i$, this further equals
$$\lambda(\mathbf{v}) =  \frac{1}{S} \sum_{i=1}^r \left (  (N-n_i) \sum_{v \in A_i} \|v \|^2\right) -\frac{2}{S} \sum_{i \neq j} A_i \cdot A_j$$
$$ = N+\frac{1}{S} \sum_{i=1}^r \left (  -n_i \sum_{v \in A_i} \|v \|^2\right) +\frac{1}{S} \sum_{i=1}^{r} 
\left|\left|\sum_{\substack{v \in A_i }} v\right|\right|^2.$$
\qed
 \end{proof}

The following theorem provides best possible drawings whenever one can draw $K_{n_1,\dots,n_r}$ on $P$ 
such that for each class $A_i$ we have $\sum_{v \in A_i} v=\vec{0}.$ This can be achieved for instance if $|A_i|$ is even for all but one of the classes, 
and for each point $v \in A_i$ in the drawing, also $-v \in A_i$, and the remaining vertex is drawn at $\mathbf{0}$. If all $|A_i|$ are even, then the 
theorem also holds under the assumption that no vertex is drawn at $\mathbf{0}$ (recall that $|P|$ is odd). 
Otherwise, %if $\sum_{v \in A_i} v\neq \vec{0}$ for all drawings, then
the best drawings are such that $\sum_{i=1}^{r} \left|\left|\sum_{\substack{v \in A_i }} v\right|\right|^2$ is minimized, 
and the drawing in the second case of the theorem only gives an approximation.

\begin{theorem}\label{thm:mainmin}
Let $\mathbf{v}$ be a straight line drawing of $K_{n_1,\dots,n_r}$ that minimizes $\lambda(\mathbf{v})$. If $n_1=n_2=\ldots=n_r$, 
then $\mathbf{v}$ minimizes $\sum_{i=1}^{r} \left|\left|\sum_{\substack{v \in A_i }} v\right|\right|^2$; 
in particular, if $\sum_{v \in A_i} v=\textbf{0}$, for all $1\leq i \leq r,$ then  
 $\lambda(\mathbf{v})=N-n_r$. %and the optimal drawings are characterized by $\sum_{v \in A_i} v=\textbf{0}$, for all $1\leq i \leq r.$
If $n_1 < n_2 < \ldots < n_r$, then $\mathbf{v}$ has the following structure:
For each $i=1,\ldots, r-1$, the union of the smallest $i$ color classes, $\bigcup_{j=1}^{i} A_j$, forms a ball centered at $\textbf{0}$.
\end{theorem}

\begin{proof}
Consider first the case when all classes $A_i$ have the same number of points $n=n_i$. Take a drawing $\mathbf{v}$. 
By Lemma~\ref{lem:minvalue},
 $$\lambda(\mathbf{v}) = N  +\frac{1}{S} \sum_{i=1}^r \left (  -n_i \sum_{v \in A_i} \|v \|^2\right)+\frac{1}{S}\sum_{i=1}^{r} 
 \left|\left|\sum_{\substack{v \in A_i }} v\right|\right|^2=N-n+\frac{1}{S}\sum_{i=1}^{r} \left|\left|\sum_{\substack{v \in A_i }} v\right|\right|^2.$$
 Then $\lambda(\mathbf{v})$ is minimized if
$\sum_{i=1}^{r} \left|\left|\sum_{\substack{v \in A_i }} v\right|\right|^2$ is minimized.
 If there are drawings $\mathbf{v}$ such that  $\sum_{v \in A_i} v=\textbf{0}$ for each class, 
 then $\sum_{i=1}^{r} \left|\left|\sum_{\substack{v \in A_i }} v\right|\right|^2=0$. 
 Since the algebraic connectivity of $K_{n,n\ldots,n}$ (with $N=r \cdot n$) is $N-n$, such a drawing is best possible.
Consider then the case $n_1 < n_2 < \ldots < n_r.$ By Lemma~\ref{lem:minvalue}, $\mathbf{v}$ minimizes $\lambda(\mathbf{v})$ if
 $\sum_{i=1}^{r} \left|\left|\sum_{\substack{v \in A_i }} v\right|\right|^2$ is minimized and $\sum_{i=1}^{r} \left(n_i \sum_{v \in A_i} \|v \|^2\right)$ is 
 as large as possible.
 Both conditions can be guaranteed at the same time.  $\sum_{i=1}^{r} \left|\left|\sum_{\substack{v \in A_i }} v\right|\right|^2$  
 can be kept small (or equal to $0$) when drawing each $A_i$ in a symmetric way around the origin. The quantity
$\sum_{i=1}^{r} \left(n_i \sum_{v \in A_i} \|v \|^2\right)$ is maximized when the smallest class $A_1$ is drawn 
such that $\sum_{v \in A_1}\|v \|^2$ is as small as possible, which is the case when the vertices of $A_1$ are drawn as close as possible to the origin;
then, in an optimal drawing the vertices of the second smallest class $A_2$ are drawn as close as possible to the origin on grid points which are not
occupied by $A_1$. In the same way, iteratively, for the $i$-th smallest class all grid points closest to the origin, that are not yet occupied by smaller classes,
are selected. This results in a drawing with concentric rings around the origin. 
\qed
\end{proof}
\begin{remark}
If some of the classes have the same number of elements, then the optimal solutions are given by a combination of the two cases of Theorem~\ref{thm:mainmin}. 
That is, several classes with the same number of elements can form one of the concentric rings in the drawing which satisfies $\sum_{v \in A_i} v=\textbf{0}$ for all color 
classes.
\end{remark}

We show next that the number of optimal drawings of $K_{n_1,\dots,n_r}$ that minimize $\lambda(\mathbf{v})$ can be exponential if some classes have the same number of elements.
For the sake of simplicity of the exposition, we show this only for the case $K_{1,2m,2m}$ and dimension $d=1$. 
The argument can be adapted to the general case.

\begin{proposition}
Let $d=1$, and $P=\{-2m,-2m+1,\ldots,2m-1,2m\}$. There exists a constant $c > 0$ such that the number $\mathcal{N}$ of straight line drawings $\mathbf{v}$ of $K_{1,2m,2m}$ on $P$  which minimize $\lambda(\mathbf{v})$ satisfies
 $c16^{m}/m^5 < \mathcal{N} < 16^{m}$.
\end{proposition}
\begin{proof}
Let $A_1,A_2,A_3$ be the classes of $K_{1,2m,2m}$, with $n_1=1$ and $n_2=n_3=2m$.
Theorem~\ref{thm:mainmin} characterizes the optimal drawings as all drawings that satisfy $\sum_{v \in A_i} v=0$. 
Then  the only vertex of class $A_1$ is drawn at position $0$ in any optimal drawing. For the upper bound, the number of such drawings is at most
${{4m}\choose{2m}} <16^{m}$, since there are at most ${{4m}\choose{2m}}$ choices for mapping the vertices of $A_2$ to $P\backslash\{0\}$, 
and then the positions of the vertices in $A_3$ are already determined. Regarding the lower bound, in order to have $\sum_{v \in A_2} v=0$, 
we must have $\sum_{v \in A_2, v < 0} -v=\sum_{v \in A_2, v > 0} v$. We may thus consider only drawings with exactly $m$ elements $v$ of $A_2$ with $v>0$. 
There are at most $\sum_{i=1}^{2m} i =2m^2+m$ different sums that can be obtained by $\sum_{v \in A_2, v > 0} v$,  and the same holds for
$\sum_{v \in A_2, v < 0} -v$. Thus, one of these sums, call it $s$, appears in at least $\frac{{{2m}\choose{m}}}{2m^2+m}$ of all the drawings 
of $\{v \in A_2, v > 0\}$, and by symmetry, the same sum $s$ appears also at least $\frac{{{2m}\choose{m}}}{2m^2+m}$ times when considering $\sum_{v \in A, v < 0} -v$.
Any drawing for which at the same time we have $\sum_{v \in A_2, v > 0} v=s$ and $\sum_{v \in A_2, v < 0} -v=s$ is an optimal drawing. 
There are at least $\left(\frac{{{2m}\choose{m}}}{2m^2+m}\right)^2=\Omega\left(\frac{16^m}{m^5}\right)$ such drawings, where we use the asymptotic estimate 
${{2m}\choose{m}} \sim  \frac{4^m}{\sqrt{\pi m}}$. Hence the lower bound follows.
\qed
\end{proof}

\section{Optimal Drawings for Maximizing $\lambda(\mathbf{v})$}

We now study drawings of $K_{n_1,\ldots,n_r}$ that maximize $\lambda(\mathbf{v})$. The following solution as a Voronoi diagram has to be considered as an approximation, due to the discrete setting and due to the given bounding box. However, the bigger the numbers $n_i$, the better the approximation to the boundary curves between adjacent Voronoi regions.

\begin{theorem}\label{thm:mainmax}
Let $\mathbf{v}$ be a straight-line drawing of $K_{n_1,\dots,n_r}$ on $P$ that maximizes $\lambda(\mathbf{v})$. If $n_1=n_2=\ldots=n_r$, then $\mathbf{v}$ defines a centroidal Voronoi diagram. If the $n_i$ are not all the same, then $\mathbf{v}$ defines a multiplicatively weighted centroidal Voronoi diagram.
\end{theorem}

\begin{proof}
We make use of the following fact: let $Q$ be an arbitrary set of $n$ points $p_1,\ldots, p_n$ in $\mathbb{R}^d$.
Let $c$ be the centroid of $Q$, $c=\frac{1}{n}\sum_{i=1}^{n} p_i.$ %Then, for any point $z \in \mathbf{R}^d$ we have, see~\cite{apostol},
 Then, see~\cite{apostol},
 \begin{equation}\label{eq:apostol}
\sum_{i=1}^{n-1}\sum_{j=i+1}^{n} ||p_i-p_j||^2=n \sum_{i=1}^{n} ||p_i - c||^2.
\end{equation}
In the case of our theorem, let $\mathbf{v}$ be a drawing of $K_{n_1,\dots,n_r}$ drawn on
\[P=\left \{ (x_1,\dots,x_d) \in \mathbb{Z}^d:-M \le x_i \le M \right \}.\]
Denote by  $c_{A_1},\ldots, c_{A_r}$ the centroids of the classes $A_1, \ldots, A_r$, respectively. Then, from Equation~(\ref{eq:def}) and $S=\sum_{v \in P}  \|v \|^2$ we get
$$\lambda(\mathbf{v})|S|= \sum_{(v,w) \in E} \| v-w\|^2 = \sum_{i<j}\sum_{\substack{v \in A_i \\ w \in A_j }}  ||v - w||^2$$
$$=\sum_{v,w \in P}||v - w||^2 - \sum_{i=1}^{r}\sum_{v,w \in A_i} ||v - w||^2$$
$$= \sum_{v,w \in P}||v - w||^2 - \sum_{i=1}^{r}n_i \sum_{v \in A_i}||v - c_{A_i}||^2,$$
where in the last equation we use~(\ref{eq:apostol}).
The quantity $\sum_{v,w \in P}||v - w||^2$ is the same for each drawing of $K_{n_1,\dots,n_r}$, and $\sum_{i=1}^{r}n_i \sum_{v \in A_i}||v - c_{A_i}||^2$ is minimized if for each class $A_i$, its vertices are drawn as close as possible to its centroid $c_{A_i}$. Then the union of the $r$ regions defined by $A_1,\ldots, A_r$ forms a centroidal Voronoi tessellation, see~\cite{du}. Note that when the $n_i's$ are different, then this is a multiplicatively weighted Voronoi diagram, see~\cite{auren}.
\qed
\end{proof}

\section{An Integer Variant of the Algebraic Connectivity}

In this section we consider drawings in dimension $d=1$ of graphs $G=(V,E)$ with $V=\{1,\ldots,N\}$, that is, drawings $\mathbf{v}$ where the vertices of $G$ 
are mapped to different points of $P= \{-\lfloor{N/2\rfloor},-\lfloor{N/2\rfloor}+1,\ldots, \lfloor{N/2\rfloor}\}.$ If $N$ is even, then in order to satisfy
the condition $\sum_{i=1}^{N} v_i = 0$ (recall the definition of (\ref{eq:def}) and the Embedding Lemma), no vertex is mapped to the origin.
We denote by $$\lambda_2^I(G)=\min \lambda(\mathbf{v}),$$ where the minimum is taken over all drawings  $\mathbf{v}$ of $G$ on $P$.
Note that when $N$ is odd, then $\lambda_2^I(G)$ is equivalent to the square of the minimum-2-sum, $\sigma_2^2(G)$ 
(recall the definition of minimum-2-sum in the introduction). %In the same way, when considering drawings of $G \times H$ with $|G|=|H|=N$ for an odd number $N$,

Continuing the investigations by Juvan and Mohar mentioned in the introduction (see~\cite{juvan2,juvan}), we are here interested in properties and bounds for $\lambda_2^I(G)$, similar in spirit to bounds and properties of $\lambda_2(G)$. 
First, the following relation, analogous to the one for $\lambda_2(G)$ from~\cite{fiedler} is obtained easily.
\begin{proposition}
If $G$ and $H$ are edge-disjoint graphs with the same set of vertices, then
$$\lambda_2^I(G)+ \lambda_2^I(H) \leq  \lambda_2^I(G \cup H ).$$
\end{proposition}
The proof is immediate from the definition of $\lambda_2^I(G)$ and the Embedding Lemma, by splitting the sum of the edge weights for $G \cup H$ into two sums of edge weights, one for $G$ and one for $H$.\\
Denote by $G+e$ the graph obtained from the graph $G$ with $N$ vertices by adding an edge $e$. It is known (see~\cite{abreu}) that $\lambda_2(G) \leq \lambda_2(G+e) \leq \lambda_2(G)+2$. We have a result in the same spirit for $\lambda_2^I(G)$:
\begin{proposition}
Denote by $G+e$ the graph obtained from the graph $G$ with $N$ vertices by adding an edge $e$.
Then $$\lambda_2^I(G)+ \frac{1}{2\sum_{i=1}^{\lfloor{N/2\rfloor}}i^2} \leq \lambda_2^I(G+e) \leq   \lambda_2^I(G)+\frac{N^2}{2\sum_{i=1}^{\lfloor{N/2\rfloor}}i^2}.$$\\
\end{proposition}
Again, the proof is immediate; adding an edge to a drawing of $G$ increases the edge weight by at least $1$ and by at most $N^2$.\\

Let us then consider the Cartesian product of graphs. Recall that the Cartesian product $G \times H$ is defined as follows: 
$V(G \times H)=V(G) \times V(H)$, and $(u,u')(v,v') \in E(G \times H)$ iff either $u=v$ and $u'v' \in E(H)$, or $u'=v'$ and $uv \in E(G)$. 
For the Cartesian product of two graphs $G$ and $H$, Fiedler~\cite{fiedler} proved the relation $\lambda_2(G \times H) = \min\{\lambda_2(G), \lambda_2(H)\}$.
The analogous relation does not hold for $\lambda_2^I(G)$;
$\lambda_2^I(G \times H)$ can be strictly larger than $\min\{\lambda_2^I(G), \lambda_2^I(H)\}$ as can be seen by the example of $C_3 \times P_2$, 
the Cartesian product of a triangle with an edge.\\
In the following, the number of vertices of a graph $G$ is also denoted by $|G|$.

\begin{proposition}\label{prop:cart}
Let $G$ and $H$ be two graphs such that $|G|$ and $|H|$ are odd numbers. Then we have
\begin{equation}
\lambda_2^I(G\times H) \leq \lambda_2^I(G)\left( \frac{|G|^2-1}{|G|^2|H|^2-1}\right)  + \lambda_2^I(H)\left( \frac{|G|^2(|H|^2-1)}{|G|^2|H|^2-1}\right)
\end{equation}
\end{proposition}

\begin{proof}
We present a drawing of $G \times H$ which attains the claimed bound.
First consider an optimal drawing $H_{opt}$ of $H$ which gives $\lambda_2^I(H)$, and then replace each vertex of $H$ by $|G|$ vertices.
More precisely, the $|H|\cdot |G|$ vertices of $G \times H$ are drawn on $P$ in such a way that we have $|H|$ consecutive copies $G_i$ of $G$ 
(each copy $G_i$ occupies an interval of $|G|$ consecutive points of $P$). Within each $G_i$ the vertices are ordered in the same way such that 
the drawing of $G_i$ is best possible (minimizing the sum of squared edge lengths); denote this drawing of $G_i$ as $G_{opt}$. 
Figure~\ref{fig:cartesian} shows such a drawing of $G\times H$ for $G=C_3$ and $H=P_3$.
\begin{figure}
	\centering
		\includegraphics{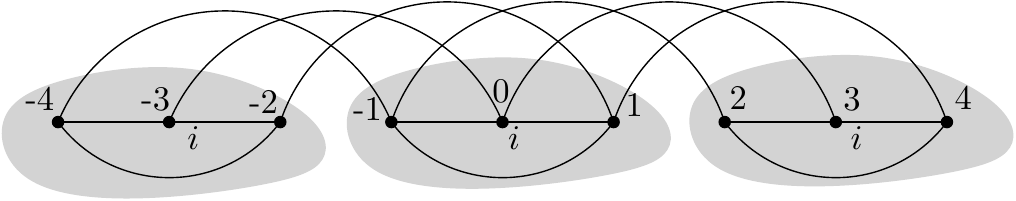}
	\caption{A drawing of the Cartesian product of cycle $C_3$ and path $P_3$. Some straight line edges are drawn as arcs for better visibility.}
	\label{fig:cartesian}
\end{figure}
From Proposition~\ref{lem:S}, with $2M+1=|G||H|$ and $d=1$, we have
$$\sum_{i \in V}||\vec{v_i}||^2=\frac{1}{12}\left(\vert G\vert^2 \vert H\vert^2 -1 \right) \left(\vert G\vert\right) \left(\vert H\vert \right).
$$
Between two consecutive copies of a vertex $i$ of $H$ there are exactly $|G|-1$ points of $P$. Then an edge $e \in H_{opt}$ with squared length $e^2$
has squared edge length $(e|G|)^2$ in our drawing of $G \times H$. We get
$$\lambda_2^I(G\times H) \leq \frac{\vert H\vert\displaystyle\sum_{e \in G_{opt}} e^2}{\frac{1}{12}\left(\vert G\vert^2 \vert H\vert^2 -1 \right)
\left(\vert G\vert\right) \left(\vert H\vert \right)} + \frac{\vert G\vert \displaystyle\sum_{e \in H_{opt}} (e |G|)^2}{\frac{1}{12}
\left(\vert G\vert^2 \vert H\vert^2 -1 \right) \left(\vert G\vert\right) \left(\vert H\vert \right)}$$

$$
 = \frac{\vert H\vert\displaystyle\sum_{G_{opt}} e^2}{\frac{1}{12}\left(\vert G\vert\right) \left(\vert H\vert \right)\left(\vert G\vert^2-1 \right)
\left( \frac{\vert G\vert^2 \vert H\vert^2 -1}{\vert G\vert^2-1}\right)} +
\frac{\vert G\vert ^3\displaystyle\sum_{H_{opt}} e^2}{\frac{1}{12}\left(\vert G\vert\right)
\left(\vert H\vert \right)\left(\vert H\vert^2 -1 \right) \left( \frac{\vert G\vert^2 \vert H\vert^2 -1}{\vert H\vert^2-1}\right)}$$
$$ = \lambda_2^I(G)\left( \frac{|G|^2-1}{|G|^2|H|^2-1}\right)  + \lambda_2^I(H)\left( \frac{|G|^2(|H|^2-1)}{|G|^2|H|^2-1}\right).$$\qed
\end{proof}

\begin{corollary}\label{cor:cartmean}
Let $G$ and $H$ be two graphs such that $|G|=|H|$ is an odd number. Then we have
$$\lambda_2^I(G\times H) \leq \frac{\lambda_2^I(G)+\lambda_2^I(H) }{2}.$$
\end{corollary}

\begin{proof}
This follows from the proof of Proposition~\ref{prop:cart},  by interchanging the role of $G$ and $H$ in the drawing, and then by summing the two inequalities.
\qed\end{proof}
\begin{corollary}
If $\lambda_2^I(G)=\lambda_2(G)$ and $|G|$ is odd, then
$$\lambda_2^I(G\times G) = \lambda_2^I(G).$$
\end{corollary}
\begin{proof}
On the one hand, $\lambda_2^I(G\times G) \geq \lambda_2(G\times G) = \lambda_2(G) = \lambda_2^I(G).$
On the other hand, by Corollary~\ref{cor:cartmean}, $\lambda_2^I(G\times G) \leq \lambda_2^I(G).$
\qed\end{proof}

The assumption of $|G|$ and $|H|$ being odd numbers in Proposition~\ref{prop:cart}
simplifies the calculations. We believe that a similar bound holds when $|G|$ or
$|H|$ are even. Indeed, the drawing for $G \times H$ explained in the proof of Proposition~\ref{prop:cart} can be optimal when $|G|$ and $|H|$ are even. We illustrate this with the hypercube and mention that its eigenvalues and eigenvectors are well known.

\begin{proposition}\label{prop:cube}
For $Q_N$, the hypercube on $N$ vertices, $\lambda_2^I(Q_N)=\lambda_2(Q_N)=2$.
\end{proposition}
\begin{proof}
To see this, note that $Q_N = Q_{N/2} \times P_2.$ 
In this case an optimal drawing of $Q_N$ can be obtained from two copies of an optimal drawing for $Q_{N/2}$
using ideas of the drawing of Proposition~\ref{prop:cart}: indeed, one can take an optimal drawing of $Q_{N/2}$ 
once shifted towards $\{1,\ldots, {N/2}\}$ (corresponding to vertices of the hypercube having $0$ in the first dimension), 
and once shifted towards $\{-{N/2},\ldots,-1\}$ (corresponding to vertices of the hypercube having $1$ in the first dimension),
and then connecting them by a matching. This drawing is similar to the one described in Proposition~\ref{prop:cart}; 
in fact, the only difference is that no vertex is mapped to the origin.\qed\end{proof}

\begin{proposition}
There are graphs $G$ with $\lambda_2^I(G \times G) < \lambda_2^I(G)$.
\end{proposition}
\begin{proof}
Let $G$ be the graph consisting of a triangle, with labels of the vertices $1,2,3$, and a path of length $2$ attached to vertex $3$; 
label these vertices $4,5$, in this order. Clearly, the function $f: V(G) \rightarrow \{-2,\ldots,2\}$ given by $f(i) = i-3$, $1 \leq i \leq 5$, 
defines an optimal drawing of $G$, yielding $\lambda_2^I(G) = \frac{8}{10}$. On the other hand, consider the drawing 
$g: V(G) \times V(G) \rightarrow \{-12,\ldots,12\}$ given as follows: $g(i,j) = -12+3(i-1)+(j-1)$ for $1 \leq i,j \leq 3$,  $g(i,j) = -3+2(i-1)+(j-4)$ 
for $1  \leq i \leq 3, 4 \leq j \leq 5$, and $g(i,j) = 3+5(i-4)+(j-1)$ for $4 \leq i \leq 5, 1 \leq j \leq 5$. 
The drawing given by $g$ gives an upper bound on $\lambda_2^I (G \times G)$, and hence $\lambda_2^I (G \times G) \leq \frac{775}{1300}< 0.6$.
\qed\end{proof}

Whereas it is obvious that for graphs $G$ with an odd number $N$ of vertices, the optimal drawings of $\lambda_2^I(G)$ and $\sigma^2_2(G)$ coincide,
this is not always the case for $N$ even.

\begin{proposition}
There exist graphs $G$ with $|G|$ even, for which the optimal drawings of $\lambda_2^I(G)$ and $\sigma^2_2(G)$ are different.
\end{proposition}
\begin{proof}
Consider the graph $G$ shown in Figure~\ref{fig:graph}. 
An optimal drawing for $\sigma^2_2(G)$ is given by ordering the vertices in the order $12354678$ or $12345678$. 
Indeed, in any drawing, the five edges incident to vertex $5$ together have squared edge length at least 
$2\cdot 1^2+2 \cdot 2^2 + 3^2$ and the other two edges have squared edge length at least $1$. 
It is easily checked that for $\lambda_2^I(G)$,  $1 8 2 7 5 3 4 6$ is a better embedding  than $12354678$ or $12345678$. 
\qed\end{proof}

\begin{figure}
	\centering
		\includegraphics{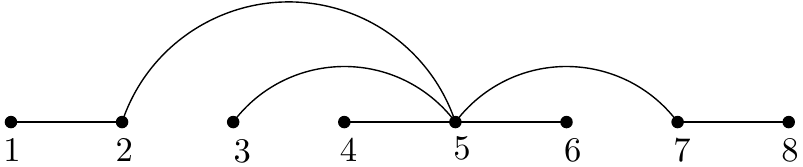}
	\caption{A graph $G$ which has different drawings for $\lambda_2^I(G)$ and for $\sigma^2_2(G)$.}
	\label{fig:graph}
\end{figure}

\section{Conclusion}
In this paper we gave drawings minimizing as well as maximizing $\lambda(\mathbf{v})$, and we analyzed properties of an integer variant of the algebraic connectivity. It would be interesting to characterize the class of graphs $G$ for which $\lambda_2(G)=\lambda_2^I(G)$.\\
{\small
\noindent \textbf {Acknowledgments.}}

We thank Igsyl Dom\'inguez for valuable discussions on early versions of this work. We also thank the reviewers for their very helpful comments.

%\vspace{-5pt}
\parpic{\includegraphics[width=0.13 \textwidth]{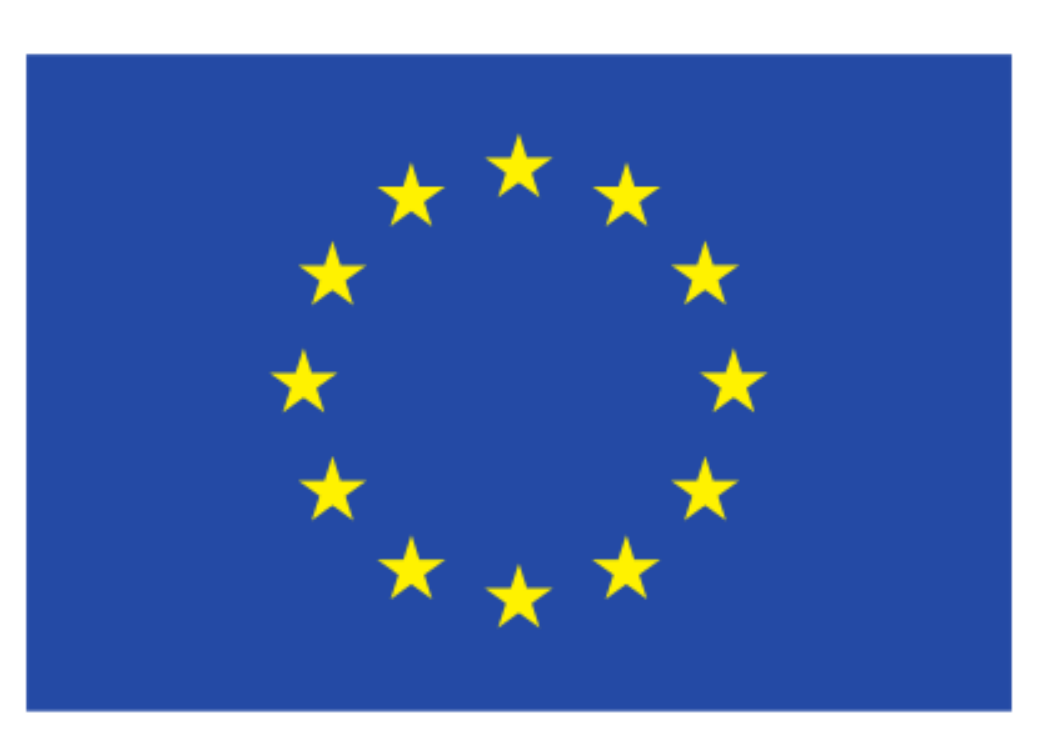}}
\noindent This project has received funding from the European Union's\\ Horizon 2020 research and innovation programme under
the\\ Marie Sk\l{}odowska-Curie grant agreement No 734922.

Clemens Huemer was supported by projects MINECO MTM2015-63791-R and Gen. \ Cat. \ DGR 2017SGR1336.
Ruy Fabila-Monroy and Carlos Hidalgo-Toscano were supported by Conacyt grant 253261.

\newpage

\section{Appendix}

In this appendix we give the omitted proofs of Section~\ref{sec:minimize}.

\propS*
%\begin{reptheorem}{lem:S}
%\[S=2d(2M+1)^{d-1}\frac{M (M+1) (2M+1) }{6}.\]
% \end{reptheorem}
  \begin{proof}
 We have that 
\[S=\sum_{v \in P}  \|v \|^2  =\sum_{(x_1,\dots,x_d) \in P} \left (x_1^2+\cdots + x_d^2 \right)=2d(2M+1)^{d-1} \sum_{i=1}^M i^2.\]
The last equation comes from counting the number of appearances of the term $i^2$ in the left hand side of the equation. This is equivalent to counting the number of times $i^2$ is equal to the
squared coordinate of a vector $v$ of $P$. Such a coordinate may be
a $i$ or a $-i$, this gives a factor of two; it has $d$ possibilities to appear as a coordinate of $v$, this gives
a factor of $d$; once the sign and position are fixed, the other $(d-1)$ coordinates can take
any one of $(2M+1)$ possible values; this gives a total of $2d(2M+1)^{d-1}$ vectors.
Finally, 
\[2d(2M+1)^{d-1} \sum_{i=1}^M i^2 = 2d(2M+1)^{d-1}\frac{M (M+1) (2M+1) }{6}.\]\qed
 \end{proof}

\propProd*
\begin{proof}
For the  case $r=2$,
$$A_1 \cdot A_2 =\sum_{\substack{v \in A_1 \\ w \in A_2 }} v \cdot w =
\left(\sum_{\substack{v \in A_1}} v\right)\cdot \left(\sum_{\substack{w \in A_2}} w\right)$$ 
$$=
\left(\sum_{\substack{v \in A_1}} v\right)\cdot \left(-\sum_{\substack{v \in A_1}} v\right) 
=-\left\|\sum_{\substack{v \in A_1}} v\right\|^2 \leq 0.$$
In the same way, $$A_1 \cdot A_2 =
\left(\sum_{\substack{w \in A_2}} w\right)\cdot \left(-\sum_{\substack{w \in A_2}} w\right) 
=-\left\|\sum_{\substack{w \in A_2}} w\right\|^2 \leq 0.$$
Summing the two equations, the result follows for $r=2$. Let then $r>2.$ For each $i \in \{1,\ldots,r\}$ we have
$$\sum_{\substack{j=1 \\ j \neq i}}^{r} A_i \cdot A_j =\sum_{\substack{j=1 \\ j \neq i}}^{r}\sum_{\substack{v \in A_i \\ w \in A_j }} v \cdot w =
\left(\sum_{\substack{v \in A_i}} v\right)\cdot \left(\sum_{\substack{w \in \cup_{j=1}^{r} A_j \\j \neq i}} w\right) 
=-\left\|\sum_{\substack{v \in A_i}} v\right\|^2 \leq 0,$$ 
where we applied the result for $r=2$. Then, when summing these $r$ equations (summing over all $i$), each term $A_i \cdot A_j$ appears exactly twice in the sum. The result follows.
\qed\end{proof}

\end{document}